\documentclass{article}
\usepackage{amsmath,amsthm,amssymb}

\usepackage{fullpage}
\usepackage{graphicx}
\usepackage{balance}
\usepackage{enumitem}
\usepackage{natbib}
\usepackage{color}
\usepackage{tikz}
\usepackage[multiple]{footmisc}
\usetikzlibrary{positioning,shadows,shapes,backgrounds,fit}

\newtheorem{theorem}{Theorem}
\newtheorem{lemma}[theorem]{Lemma}
\newtheorem{proposition}[theorem]{Proposition}
\newtheorem{obs}[theorem]{Observation}

\renewcommand{\epsilon}{\varepsilon}
\newcommand{\out}{\mathcal{S}}
\newcommand{\vett}[1]{\mathbf{#1} }
\renewcommand{\b}{\vett b}
\newcommand{\bi}{\b_{-i}}
\newcommand{\lmin}{t_{\inf}}
\newcommand{\lmax}{t_{\sup}}
\newcommand{\bp}{\mathbf{b}}
\newcommand{\tp}{\mathbf{t}}
\newcommand{\calR}{\mathcal{R}}
\newcommand{\bsig}{\boldsymbol{\sigma}}

\newcommand{\real}{\mbox{$\mathbb{R}$}}
\newcommand{\M}{\mathcal{M}}

\DeclareFontFamily{U}{mathx}{\hyphenchar\font45}
\DeclareFontShape{U}{mathx}{m}{n}{
	<5> <6> <7> <8> <9> <10>
	<10.95> <12> <14.4> <17.28> <20.74> <24.88>
	mathx10
}{}
\DeclareSymbolFont{mathx}{U}{mathx}{m}{n}
\DeclareMathSymbol{\bigtimes}{1}{mathx}{"91}

\title{Probabilistic Verification for Obviously Strategyproof Mechanisms}
\author{Diodato Ferraioli \and Carmine Ventre}
\date{}

\begin{document}

\raggedbottom

\maketitle

\begin{abstract}
Obviously strategyproof (OSP) mechanisms maintain the incentive compatibility 
of agents that are not fully rational.
They have been object of a number of studies since their recent definition.
A research agenda, initiated in \cite{FerraioliV17}, is to find a small (possibly, the smallest) set 
of conditions allowing to implement an OSP mechanism.
To this aim,
we define a model of probabilistic verification wherein agents are caught misbehaving with a certain probability, and show how OSP mechanisms can implement every social choice function at the cost of either imposing very large fines for lies or verifying a linear number of agents.
\end{abstract}

\section{Introduction}
Will people strategize against an incentive-compatible mechanism? The answer depends on whether they will \emph{understand} that doing so is against their own interest and, ultimately, on their rationality and cognitive skills. This question has often been raised in literature (see, e.g., \citep{sandholm2003sequences,FerraioliVA15}) and much of the recent research in mechanism design is motivated by this question.
Several definitions for ``simple'' mechanisms have been recently given in literature:
posted-price mechanisms and variants \citep{chawla2010multi,babaioff2014simple,adamczyk2015sequential}, Bulow-Klemperer-like auctions \citep{hartline2009simple}, verifiably truthful mechanisms \citep{branzei2015verifiably}.
This quest for the right definition for simple mechanisms
culminated with the introduction of \emph{obviously strategyproof} (OSP) mechanisms \citep{li2015obviously}.

Obvious strategyproofness focuses on how a mechanism is executed (e.g., English auction vs. sealed bid second price auction), and requires that whenever an agent takes an action during the execution of the mechanism, the ``truthful behavior'' must be dominant for that agent, even if no reasoning is done about the future actions of the other agents.
This concept is motivated by the experimental evidence that some mechanism implementations (e.g., clock auction) are easier to understand than other theoretically equivalent ones (e.g., sealed-bid auction). OSP mechanisms have also solid theoretical foundations: they are the only ones that preserve the incentive-compatibility of agents who lack contingent reasoning skills \citep{li2015obviously} and they satisfy a natural generalization of standard decision theory axioms \citep{zhang2017partition}.

This concept has attracted a considerable amount of recent work \citep{ashlagi2015no,BadeG17,FerraioliV17,pycia2016obvious} that mainly focuses on the properties and the limitations of these mechanisms.
Of particular interest for our study are the results proved by \cite{FerraioliV17} showing that OSP mechanisms cannot have good approximation guarantees for machine scheduling and facility location,
two canonical optimization problems studied in the area.
However, monetary transfers are sufficient for the existence of optimal OSP mechanisms when the designer can ``monitor'' all agents
(meaning that a lying agent is artificially made to receive an utility
that is the worst between the one computed according to her true type
and the one computed as if her type coincided with her bid).
Since money is undesirable in many applications
(cf. the vast literature on approximate mechanism design without money initiated by \cite{PT09})
our main aim here is to understand whether there are novel ways for the designer
to exert control over the agents that
can reconcile approximation and OSP mechanisms
with limited or no transfer of money.

\paragraph{Our contribution.}
We introduce a model of \emph{probabilistic verification} wherein the mechanism designer has access to a (potentially faulty) verification device that she can use at runtime to check whether an agent has lied or not. The device will catch the lie of the checked agent with certainty, or with a certain probability if faulty.  For example, whenever the type $t$ of an agent is her location on the real line (as in facility location), the designer can use a GPS logger to check whether the agent location is the same as her reported type $b$. In our terminology, such a tool is faulty if its reading $t'$ of $t$ is subject to some measurement error $\delta$ and, therefore, the agent would be caught only if $|b-t'|>\delta$; more generally, one could imagine different tools that make mistakes in their measurements with some probability rather than in range (e.g., it gets better as the difference between reported and real type increases). This notion generalizes and combines the different notions of verification introduced in related literature, see, e.g., \citep{ec12,PenVen09} -- a longer discussion on this aspect can be found in Section \ref{sec:preli}.
With respect to monitoring, the mechanism designer has in our probabilistic model a more general way to define fines for lying agents, whilst, on the other hand, might have a faulty verification device.

We begin by studying what we call the \emph{full probabilistic verification model}, wherein every agent is verifiable.
and therefore there is a non-null probability of catching lies.
We prove that, in this setting, it is possible to obtain an OSP mechanism for every specific problem of interest;
we essentially show that we can always define verification probabilities and corresponding fines to make any lie obviously dominated.
On the technical level, we show that there is a trade-off between the kind of verification device needed (i.e., the verification probabilities)
and the amount of fines imposed to lying agents that are caught. Interestingly, our results also imply that we can set the fines
high enough to only verify a constant number of agents in expectation.

The result above has two main limitations: firstly, it
requires all the agents to be verifiable -- this might be impossible in some contexts (e.g., not all the agents might have been equipped with a GPS logger); secondly, the fines might have a value that makes their enforceability doubtful (this is certainty the case for the fines needed to guarantee a low number of verified agents).
We therefore look at the \emph{partial probabilistic verification model}, where for some agents we cannot use any verification (and so the combination of fines and probabilistic checks will not make lying obviously dominated). In the main technical contribution of this work, we prove that there is a problem such that all $\epsilon$-OSP mechanisms (i.e., agents will not deviate for small gains $\epsilon$) that solve this problem need to verify in expectation a linear number of agents. We focus on the well studied \emph{public project problem} \citep{jackson1992implementing} and identify a small domain for which ``many'' agents need to be verified by any $\epsilon$-OSP mechanism that solves the problem.
This key result is then leveraged to define instances where every $\epsilon$-OSP mechanism for this problem needs to verify a linear number of agents.
Furthermore, this result is extended to $\epsilon$-OSP mechanisms that solve the problem only asymptotically (i.e., the solution returned by the mechanism gets closer to the right one as the number of agents goes to infinity) and to $\epsilon$-OSP Bayesian mechanisms (i.e., mechanism designer knows a prior on the distribution of agents' valuations); incidentally, as far as we know, this is one of the first results on OSP mechanisms in the Bayesian setting -- it is interesting to observe that, in this context, the prior does not give any advantage to the designer of OSP mechanisms.

We finally prove that this result is essentially tight for $\epsilon$-OSP \emph{in expectation} mechanisms that implement a social choice function \emph{asymptotically}. In detail, we connect OSP with \emph{differential privacy} and show how the exponential mechanism of \cite{diff} can be implemented with partial probabilistic verification to become $\epsilon$-OSP in expectation and verify $n - o(n)$ agents. Although the proofs basically follow the known ones, we regard this result interesting for three main reasons (in addition to showing the tightness of the lower bound). Firstly, it shows how, through verification, differential privacy can be related to OSP just like truthfulness. Secondly, the mechanism becomes implementable not just in our probabilistic framework but also with selective verification, a widely studied concept studied by economists and computer scientists (cf, e.g.,  \citep{FotakisTZ16}). Thirdly, it is, to the best of our knowledge, the first instance of a mechanism whose obvious strategyproofness is guaranteed in expectation over the random choices of the mechanism. We actually need to extend the definition of \cite{li2015obviously} to allow for this --- we regard such a definition as an important conceptual contribution of this work.

\section{Preliminaries}\label{sec:preli}
A mechanism design setting is defined by a set $N$ of $n$ \emph{selfish agents} and a set of allowed \emph{outcomes} $\out$.
Each agent $i$ has a \emph{type} $t_i \in D_i$, where $D_i$ is called the \emph{domain} of $i$.
The type $t_i$ is usually assumed to be \emph{private knowledge} of agent $i$.
Sometimes, it is possible that a distribution $\Delta_i$ from which $t_i$ is drawn is known \emph{to the designer}: when this is the case, we say that we are in a \emph{Bayesian} setting.\footnote{Observe that in our Bayesian setting, the players are \emph{not} aware of the prior. Thus, the OSP-ness of our mechanisms are prior independent. A different meaning is sometimes used in literature for Bayesian mechanism design, wherein the prior also modifies the definition of truthfulness.}
In the rest of the paper, we will always refer to the classical non-Bayesian setting, unless differently specified.
Each selfish agent $i$ is equipped with a \emph{valuation function} $v_i \colon D_i \times \out \rightarrow \real$.
For $t_i \in D_i$ and $x \in \out$, $v_i(t_i, x)$ is the valuation that agent $i$ has for outcome $x$ when her type is $t_i$.
We will often use $t_i(x)$ as a shorthand for $v_i(t_i, x)$.
We say that the domain $D_i$ of agent $i$ is \emph{bounded} if $t_i(x) \in [\lmin, \lmax]$ for all $i$, $t \in D_i$, $x \in \out$.

A \emph{mechanism} is a process for selecting an outcome $X \in \out$.
Specifically, we can model it as an \emph{extensive game form with consequences in $\out$},
represented by a tuple $\M=(H, \prec, P, A, \mathcal{A}, \delta_c, (\mathcal{I}_i)_{i \in N}, g)$,
where:
\begin{itemize}
 \item $H$ is a set of \emph{histories} and $\prec$ is a \emph{partial order} on $H$ such that $(H, \prec)$ form an \emph{arborescence} (i.e., a directed tree with all edges pointing away from the root) of finite depth. We denote with $h_{\emptyset}$ the \emph{empty history} at the root of the arborescence, and with $Z$ the set of \emph{terminal histories}, that appear as leafs of the arborescence. Moreover, we denote as $\sigma(h)$ the set of successors of history $h$ in the arborescence;
 \item $P \colon H \setminus Z \rightarrow N \cup \{c\}$ is the \emph{player function}, that determines at each history of the mechanisms which player will take the next action, where $c$ denotes a \emph{chance} move, that is an action taken by the mechanism without interacting with any agent;
 \item $A$ is the set of all available \emph{actions} that can be taken by the agents or by chance;
 \item $\mathcal{A} \colon H \setminus Z \rightarrow A$ is the \emph{action function} that assigns to each history the set of actions that $P(h)$ can take. We require that $\mathcal{A}(h)$ is one-to-one with $\sigma(h)$, i.e., to different histories immediately reachable from $h$ there must correspond different actions taken by the agent $P(h)$ (e.g., if at history $h$ the agent $P(h)$ plays and her available actions are only saying yes or no, then $h$ has only two successors, $h'$ and $h''$, and a ``yes'' answer must lead to history $h'$ only, while a ``no'' answer must lead to history $h''$ only);
 \item $\delta_c$ is the \emph{chance} function, that is a probability distribution on the actions $\mathcal{H}$ available at histories $h$ such that $P(h) = c$.
 It serves to model random coin tosses of the mechanism;
 \item $\mathcal{I}_i$ is the collection of \emph{information sets} of $i$, that is 
 a partition of $\{h \colon P(h)=i\}$, i.e., the set of histories in which $i$ takes an action, such that $\mathcal{A}(h) = \mathcal{A}(h')$ for every $h, h'$ in the same cell of the partition, and each action is available at only one information set; formally, for every pair of information sets $I_i, I_i' \in \mathcal{I}_i$, if $a \in \mathcal{A}(I_i)$, then $a \notin \mathcal{A}(I'_i)$, where $\mathcal{A}(I_i) = \mathcal{A}(h)$ for any $h \in I_i$. Information sets are used to model incomplete information by agents: whenever agent $i$ who is called to take an action does not know if she is at history $h$ or at history $h'$ (e.g., because multiple agents are taking actions at the same time), then we put $h$ and $h'$ in the same information set $I_i$. Clearly, since these two histories are indistinguishable, then it must be the case that the set of actions available in them must be exactly the same. Moreover, for two distinct information sets it is without loss of generality to assign different labels the actions available at these information sets;
 \item $g \colon Z \rightarrow \out$ is the \emph{outcome function}, and assigns to each terminal history an outcome.
\end{itemize}

Given a mechanism $\M$, a \emph{strategy} $S_i \colon \mathcal{I}_i \rightarrow A$ of agent $i$ for this game
assigns at each information set of agent $i$ an action among the ones available at that information set,
i.e., $S_i(I_i) \in \mathcal{A}(I_i)$ for every $I_i \in \mathcal{I}_i$.
Then given a realization $d_c$ of $\delta_c$, the strategy $S_i$ of agent $i$,
and strategies $S_{-i} = (S_{1}, \ldots, S_{i-1}, S_{i+1}, \ldots, S_n)$ for every remaining player,
we denote the corresponding terminal history of the game as $z(d_c,S_i,S_{-i})$
and the outcome of the mechanism as $\M(d_c,S_i,S_{-i}) = g(z(d_c,S_i,S_{-i}))$.
Then, a strategy $S_i$ is \emph{$\varepsilon$-weakly dominant} for a player of type $t_i$ if, {for every realization $d_c$}
$$
t_i(\M(d_c,S_i,S_{-i})) \geq t_i(\M(d_c,S'_i,S_{-i})) - \varepsilon
$$
for every $S'_i$ and every $S_{-i}$.
Strategy $S_i$ is \emph{$\varepsilon$-weakly dominant} {in expectation} for a player of type $t_i$ if
$$
 E_{d_c \sim \delta_c}[t_i(\M(d_c,S_i,S_{-i}))] \geq E_{d_c \sim \delta_c}[t_i(\M(d_c,S'_i,S_{-i}))] - \varepsilon
$$
for every $S'_i$ and every $S_{-i}$.

Given two strategies $S_i$ and $S'_i$, their \emph{earliest point of departure} $\alpha(S_i, S'_i)$ is the collection of information sets such that
$S_i(I_i) \neq S'_i(I_i)$ but there is $S_{-i}$ and $d_c$ such that $I_i$ is reached both when the game is played with strategies $(S_i, S_{-i}, d_c)$ and with strategies $(S'_i, S_{-i}, d_c)$. In this case, we say that $(S_{-i},d_c)$ witnesses $I_i$ being in $\alpha(S_i, S'_i)$. Given an information set $I_i \in \alpha(S_i, S'_i)$, we let $W(I_i, S_i, S'_i)$ denote the set of pairs $(S_{-i}, d_c)$ that witness $I_i$ being in $\alpha(S_i, S'_i)$.
Then, a strategy $S_i$ is \emph{$\varepsilon$-obviously dominant} for a player of type $t_i$ if for every $S'_i$ and every $I_i \in \alpha(S_i, S'_i)$
$$
 \inf_{(S_{-i},d_c) \in W(I_i, S_i, S'_i)} t_i(\M(d_c,S_i,S_{-i})) \geq \sup_{(S_{-i},d_c) \in W(I_i, S_i, S'_i)} t_i(\M(d_c,S'_i,S_{-i})) - \varepsilon.
$$
That is, a strategy $S_i$ is $\varepsilon$-obviously dominant if whenever agent $i$ must take an action, then, by taking the one suggested by $S_i$ she achieves an utility that is at most $\varepsilon$ far away from what she would achieve by taking a different action, \emph{regardless of what other players will do from that time on}.
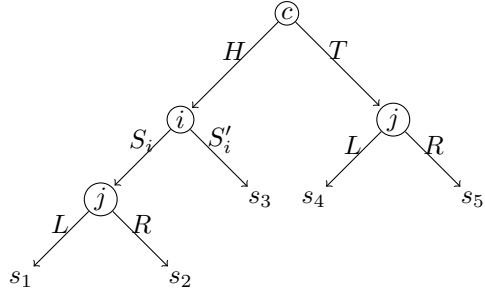
\begin{figure}[htb]
 \label{fig:expOSP}
		\centering
		\begin{tikzpicture}[shorten >=1pt,node distance=1.5cm,inner sep=1pt]
		\tikzstyle{place}=[circle,draw]
		\tikzstyle{placew}=[rectangle]
		\node[place] (x)  		     {$c$};
		\node[place] (y)   [below left of=x, node distance=2cm]  {$i$};
		\node[place] (a)   [below right of=x, node distance=2cm]  {$j$};
		\node[placew] (e)   [below left of=a]  {$s_4$};
		\node[placew] (f)   [below right of=a]  {$s_5$};
		\node[place] (z) [below left of=y]  {$j$};
		\node[placew] (b)   [below right of=y]  {$s_3$};
		\node[placew] (c)   [below left of=z]  {$s_1$};
		\node[placew] (d)   [below right of=z]  {$s_2$};
		\path[->] (x) edge node [label={$H$}] {} (y);
		\path[->] (x) edge node [label={$T$}] {} (a);
		\path[->] (a) edge node [label={$L$}] {} (e);
		\path[->] (a) edge node [label={$R$}] {} (f);
		\path[->] (y) edge node [label={$S_i$}] {} (z);
		\path[->] (y) edge node [label={$S'_i$}] {} (b);
		\path[->] (z) edge node [label={$L$}] {} (c);
		\path[->] (z) edge node [label={$R$}] {} (d);
		\end{tikzpicture}
 \caption{This arborescence describes the following mechanism: it first tosses a coin; if the coin lands on head, then it interacts only once with agent $i$, otherwise it returns an outcome from $\out$ without any interaction with agent $i$. Hence, in $\alpha(S_i, S'_i)$ there is only one information set, that in turn contains a single history $h$, corresponding to the single node labeled $i$ in the arborescence. However, this node can be only reached if the coin lands on head. Thus, if agent $i$ interacts with the mechanism, then she knows that realization $d_c = T$ is impossible, and hence she can evaluate her utility by looking only at outcomes reachable in the arborescence rooted at history $h$.}
\end{figure}

We next extend these concepts to define obvious dominance in expectation. We let $\widehat{W}(I_i, S_i, S'_i)$ denote the set $S_{-i}$ for which there is $d_c$ such that the pair $(S_{-i}, d_c)$ witnesses $I_i$ being in $\alpha(S_i, S'_i)$. Moreover, let $\delta_c \mid (I_i, S_i, S'_i, S_{-i})$ be the probability on the realizations $d_c$ of $\delta_c$ conditioned on the pair $(S_{-i}, d_c)$ witnessing that $I_i$ is in $\alpha(S_i, S'_i)$. {We say that a strategy $S_i$ is \emph{$\varepsilon$-obviously dominant in expectation} for a player of type $t_i$ if for every $S'_i$ and every $I_i \in \alpha(S_i, S'_i)$
$$
 \inf_{S_{-i} \in \widehat{W}(I_i, S_i, S'_i)} E_{d_c \sim \widehat{\delta_c}}[t_i(\M(d_c,S_i,S_{-i}))] \geq \sup_{S_{-i} \in \widehat{W}(I_i, S_i, S'_i)} E_{d_c \sim \widehat{\delta_c}}[t_i(\M(d_c,S'_i,S_{-i}))] - \varepsilon
$$
where $\widehat{\delta_c}$ is used as a shorthand for $\delta_c \mid (I_i, S_i, S'_i, S_{-i})$.}  It is important to note that the conditioning of $\delta_c$ is needed; as showed in Figure~\ref{fig:expOSP}, there may be some randomized choices of the mechanism that do not allow to reach a certain information set $I_i$. 

A mechanism designer would aim for each agent to select one specific strategy among the available ones when her type is $t_i$
(e.g., the designer would like that agents bid truthfully in sealed bid auctions,
and that agents leave only when the current price is above their type in English auctions).
We denote with $\mathbf{S}_i(t_i)$ the strategy that the designer would expect agent $i$ to play when her type is $t_i$,
and we say that $\mathbf{S}_i(t_i)$ is \emph{signalling} $t_i$.
Then, a mechanism is \emph{$\varepsilon$-strategy-proof} ($\varepsilon$-SP) if for every $i$ and every $t_i$, strategy $\mathbf{S}_i(t_i)$ is $\varepsilon$-weakly dominant\footnote{{For randomized mechanism, this concept is usually referred to as \emph{universally $\varepsilon$-strategy-proof}.}}. {If strategies $\mathbf{S}_i(t_i)$ are only $\varepsilon$-weakly dominant in expectation, then we say that the mechanism is $\varepsilon$-SP in expectation.}
A mechanism is instead \emph{$\varepsilon$-obviously strategy-proof} ($\varepsilon$-OSP) if for every $i$ and every $t_i$, strategy $\mathbf{S}_i(t_i)$ is $\varepsilon$-obviously dominant.
{As above, if strategy $\mathbf{S}_i(t_i)$ is only $\varepsilon$-obviously dominant in expectation, then the mechanism is said to be $\varepsilon$-OSP in expectation.}
Finally, a mechanism is OSP if it is $0$-OSP.
Observe that if a mechanism is obviously strategy-proof {(in expectation)}, then it is also strategy-proof {(in expectation)}.

With a slight abuse of notation, we write below $\M({d_c},(\mathbf{S}_i(b_i))_i)$ as $\M({d_c},\b)$, $\b=(b_1, \ldots, b_n)$, to link the outcome computed by the mechanism to the types of the agents more clearly.
{Moreover, for deterministic mechanisms we omit $d_c$.}
Given a \emph{social choice function} $f \colon D \rightarrow \out$, where $D = D_1 \times \cdots \times D_n$ is the set of type profiles $\b$, a mechanism $\M$ is said to \emph{implement} $f$ (in expectation, respectively) if $\M(\b) = f(\b)$ ($E_{{d_c \sim \delta_c}}[\M({d_c,}\b)] = f(\b)$, respectively) for every $\b$.
$\M$ is instead said to \emph{implement $f$ asymptotically} (in expectation, respectively) if
$\lim_{n\rightarrow \infty}\M(\b) = f(\b)$ ($\lim_{n\rightarrow \infty}E_{{d_c \sim \delta_c}}[\M({d_c,}\b)] = f(\b)$, respectively).

\paragraph{Probabilistic verification.}
We introduce a general model of probabilistic verification, inspired by \citep{ec12,PenVen09}.

Fix $i$ and $\bi$. Let $t$ and $t'$ denote the true and reported type of agent $i$, respectively. A \emph{mechanism with probabilistic verification} $\M$ catches agent $i$ lying with probability $(1-p^{i}_{t',t}(\bi))$ and punishes the agent caught lying with a fine $F^i_{t',t}(\bi) >0$. (So $p^i_{t',t}(\bi)$ denotes the probability that the verification has \emph{not} worked -- clearly, $p^i_{t,t}(\bi) =0$.) We drop $i$ from the notation when this is clear from the context.
We follow the literature and assume that verification occurs after the outcome has been computed and executed; therefore verification probabilities $p^i_{t',t}(\bi)$ and fines $F^i_{t',t}(\bi)$ are actually also functions of $\M(t', \bi)$. Moreover, we assume that when a mechanism with probabilistic verification $\M$ catches agent $i$ lying it acquires knowledge of $t(\M(t', \bi))$.
In some applications, the verification might actually be more powerful and reveal more information about $t$ (see Section \ref{sec:full} for an example).

Except for the fines, the mechanism does not use any other form of transfers. In this sense, our research extends the literature on mechanisms that trade money with verification to ensure incentive-compatibility (see, e.g., \citep{FotakisKV17,FerraioliSV16}).
When misreporting her type to a mechanism with probabilistic verification, agent $i$ will then have a valuation
$$
t(\M(t', \bi)) - (1-p^i_{t',t}(\bi)) F^i_{t',t}(\bi).
$$
There are two complementary interpretations of this formula, depending on the power of the verification device that the mechanism designer has at her disposal.
The first assumes that the verification device is faulty (e.g., subject to measurement errors) and even if $i$ is verified there is a chance that depends on type, bid and what the others reported that she is not caught (e.g., error might depend on the ``distance'' between $t$ and $t'$). The second, instead, is closer to the selective verification of \cite{FotakisTZ16} in that the device is faultless and once an agent is selected to be verified by the mechanism she will be fined with certainty if she lied. Naturally, as the mechanism has no knowledge of $t$, the probability with which the mechanism selects agent $i$ for verification can only depend, in this case, on her identity, report and bids of the others but not on her type $t$; i.e., in the above formula $p^i_{t',t}(\bi)$ reads $p^i_{t'}(\bi)$.

We will consider two different categories of mechanisms with probabilistic verification: the full model wherein all the agents are verifiable, so that we can define $p^i_{t',t}(\bi) \in [0,1]$ for every tuple $(i, t, t',\bi)$, and the partial model wherein there exists at least one agent $i$ that is not verifiable, that is, for which we require $p^i_{t',t}(\bi)=1$ for every $\bi$ and every $t, t'$ with $t \neq t'$. The non-verifiable agents might be given a priori (e.g., agents without GPS loggers) or be determined by the designer's limited resources (e.g., with $k$ out of $n$ loggers available to allocate, there would remain $n-k$ non-verifiable agents).

Our model generalizes both the aforementioned models of \cite{ec12} and \cite{PenVen09}. The latter only focuses on degenerate \emph{ex-post} verification probabilities (i.e., verification that depends on the outcome computed by the mechanism) defined, for example, as follows:
\[
p^i_{t',t}(\bi) = \left\{ \begin{array}{ll}
1 & \textrm{if } t(\M(t', \bi)) \geq t'(\M(t', \bi))\\
0 & \textrm{otherwise}
\end{array}\right.
.
\]
Moreover, the only fines imposed in \citep{PenVen09}, for mechanisms with money, is the payment that the mechanism would have awarded to the player caught lying. Subsequent literature, such as, \citep{KV15,FotakisKV17,FerraioliSV16}, instead assumes that the fines are large enough to make a verified lie very costly to the agents. We here maintain the ex-post nature of verification, but we generalize both the definitions of fines and verification probabilities. The model in \citep{ec12} is closer to ours in that probabilities and fines are not degenerate and restricted, respectively. However, \citet{ec12} use \emph{ex-ante} verification and fines, that is, verification probabilities and fines that do not depend on the outcome computed by the mechanism. In particular, for ex-ante verification, the probabilities are type-only, i.e., $p^i_{t',t}(\bi)$ only depends on $t$ and $t'$ and not on $\M(t', \bi)$; furthermore, in general, the verification step does not reveal any further information on the type/valuation of a lying agent. Our example of a GPS logger in the introduction fits this type-only model as the probability of being caught only depends on the distance between $t$ and $t'$ but not on where the mechanism locates the facility and how much the agent likes a certain outcome. On the one hand, this model is a special case of ours as we can always disregard the outcome in our definition of the $p$'s and the $F$'s. On the other hand, the ex-ante paradigm is pretty limiting. For example, a mechanism with ex-ante fines and ex-ante verification might need higher fines to ensure OSP (cf. our definition of fines in Propositions \ref{prop:fixProb} and \ref{prop:full_model}).  

\section{Full probabilistic verification}\label{sec:full}
In this section we prove that full probabilistic verification is very powerful.
Specifically, let $C$ be the random variable that denotes the expected number of agents that the mechanism is able to catch lying,
i.e. $C = \sum_{i=1}^n C_i$, where $C_i = 1$ if agent $i$ is caught lying (thus with probability $1-p^i_{t',t}(\bi)$), and $0$ otherwise.
We then prove the following theorem.
\begin{theorem}
\label{thm:positive}
 If the domains of agents are bounded, then for every social choice function $f$ there is an OSP mechanism with full probabilistic verification that implements $f$ and verifies in expectation only a constant number of agents, i.e., $E[C]=O(1)$.
\end{theorem}
Hence, there are social choice functions for which an OSP mechanism is implementable in the full probabilistic verification model but not implementable in the standard model with money (e.g., facility location \citep{FerraioliV17}).

Unfortunately, the mechanism that proves Theorem~\ref{thm:positive} needs very large fines. However, we prove that full probabilistic verification still turns out to be a powerful tool even if large fines are not available.
In particular, we observe that the proof of Theorem~\ref{thm:positive} gives 
a trade-off between fines and the number of verified agents.
Hence, one may be able to work with lower fines, by simply having more accurate verification
(in a sense, we can reduce fines only if we spend more for our verification tools).
Hence, we can show that if a lower bound on the fines is given, it is possible, under opportune conditions,
to compute verification probabilities such that the resulting mechanism with full probabilistic verification is OSP.
We also consider the opposite direction. That is, we assume that verification probabilities are given,
and investigate the lowest fines that one needs to set in order to have an OSP mechanism.
Hence, our results actually consider all the following possible cases:
(i) we are given bounded fines, and we need to decide how faulty we can allow the verification device to be in order to have an OSP mechanism (cf. Lemma~\ref{lem:fixFines}); (ii) the faultiness of the verification device is given, and we need to set fines for the mechanism to be OSP (cf. Proposition~\ref{prop:fixProb}); (iii) we require that the expected number of verified agents is limited, and we design a corresponding OSP mechanism, by computing both the necessary faultiness of the verification device and the fines (cf. Theorem~\ref{thm:positive}).

Whereas this bound can be in general very large,
we show that there are some applications (including, e.g., facility location) in which very large fines may be unnecessary.

\paragraph{OSP mechanism with few verified agents.}
For every $i$ with true type $t$, let $F^i_{\arg} = \arg_{t,t',\bi} F_{t',t}^i(\bi)$, with $\arg \in \{\min, \max\}$. 
\begin{lemma}
\label{lem:fixFines}
 For every social choice function $f$ and fines $F^i_{t',t}(\bi)$ such that $F_{\min}^i \geq \lmax - \lmin$ for all $i$,
 let $\M_F$ be the mechanism with full probabilistic verification that
 requires agents to directly reveal their type concurrently,
 implements $f$, uses fines $F^i_{t',t}(\bi)$,
 and sets $p_{t',t}^i(\bi) = \frac{\lmin-\lmax+F^i_{\min}}{F^i_{\max}}$.
 Then
 $\M_F$ is OSP.
\end{lemma}

\begin{proof}
First observe that the mechanism is well defined since, by hypothesis,
 $F_{\min}^i \geq \lmax - \lmin$ implies $p_{t',t}^i(\bi) \geq 0$ for every $(i,t,t',\bi)$. Moreover, since $\lmin < \lmax$ and $F^i_{\min} \leq {F^i_{\max}}$ then we get  $p_{t',t}^i(\bi) < 1$ for every $(i,t,t',\bi)$.

 For OSP-ness, fix agent $i$. Since $p_{t',t}^i(\bi') < 1$ for each $(t,t',\bi')$, then
\begin{align*}
   1-p_{t',t}(\bi') & = 1 - \frac{\lmin-\lmax+F^i_{\min}}{F^i_{\max}}\\
   & = \frac{(F^i_{\max} - F^i_{\min}) + (\lmax-\lmin)}{(F^i_{\max} - F^i_{\min}) + F^i_{\min}}\\
   & \geq \frac{\lmax-\lmin}{F^i_{\min}} \geq \frac{\lmax-\lmin}{F^i_{t',t}(\bi')},
  \end{align*}
where the first inequality follows since, for every $C \geq 0$, $\frac{C+x}{C+y} \geq \frac{x}{y}$ whenever $x \leq y$. The lemma then follows since, for every $\bi$,
\begin{align*}
 	t(\M_F(t, \bi)) & \geq \lmin = \lmax - (1-p_{t',t}(\bi')) \frac{\lmax-\lmin}{1-p_{t',t}(\bi')}\\
 	& \geq \lmax - (1-p_{t',t}(\bi')) F^i_{t',t}(\bi')\\
 	& \geq t(\M_F(t', \bi')) - (1-p_{t',t}(\bi')) F^i_{t',t}(\bi'). \tag*{\qed}
 	\end{align*}
\let\qed\relax
\end{proof}
Lemma~\ref{lem:fixFines} allows us to understand how the choice of $F^i_{\min}$ and $F^i_{\max}$ changes the probabilities.
In particular, how can we minimize the amount of verification needed or, equivalently, maximize $p_{t',t}^i$? 
The first observation is that the probability $p_{t',t}^i$ is higher when $F^i_{\min} = F^i_{\max}=F$.
The subsequent observation is then that the this probability quickly grows to $1$ as $F$ increases
from $\lmax - \lmin$ 
(simply look at the derivative of $(-\lmax + \lmin + x)/x$):
this shows that according to the choice of fines, there is a sort of all-or-nothing verification (see Figure \ref{fig:01ver}).

\begin{figure*}[t]
	\includegraphics[scale=0.21]{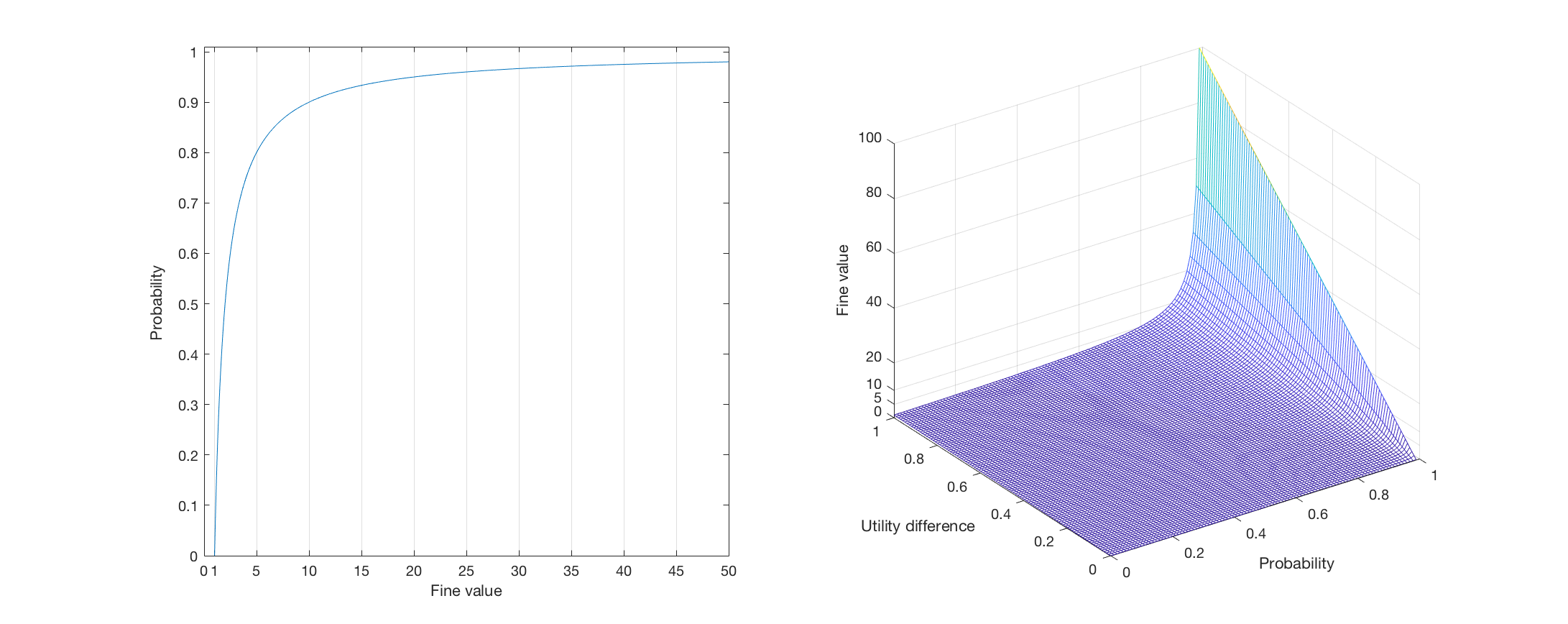} 
	\caption{Assume that $\lmin=0, \lmax=1$. Left figure shows that the value of $p^i_{t',t}(\bi)$ quickly grows to $1$ as the value of the fines increases: a fine smaller than $1$ does not help at all and we need faultless verification. Right figure shows the value of the fines as a function of the difference  $t(\M(t', \bi))-\lmin$ and $p_{\max}^i$.}\label{fig:01ver}
\end{figure*}

We are now ready to prove Theorem~\ref{thm:positive}.
\begin{proof}[\underline{Proof of Theorem~\ref{thm:positive}}]
Set, for each $(i,t,t',\bi)$, $F^i_{t',t}(\bi)=\gamma(\lmax-\lmin)$ for some $\gamma > 1$ that will be fixed later,
and let $\M_F$ be the mechanism with full probabilistic verification that
requires agents to  to directly reveal their type concurrently, implements $f$, sets $p_{t',t}^i(\bi) = \frac{\lmin-\lmax+F^i_{\min}}{F^i_{\max}}$,
and uses fines $F^i_{t',t}(\bi)$.
Since $F_{\min}^i = F^i_{\max} = F^i_{t',t}(\bi) \geq \lmax - \lmin$, then, according to Lemma~\ref{lem:fixFines}, $\M_F$ is OSP.

Moreover, given $\bi$ and fines as above, this mechanism verifies
\begin{align*}
\sum_{i=1}^n (1 - p_{b_i,t_i}^i(\bi)) & = n - \sum_{i=1}^n \frac{\lmin-\lmax+F^i_{\min}}{F^i_{\max}}\\
 & = n - \sum_{i=1}^n \left(1 - \frac{1}{\gamma}\right) = \frac{n}{\gamma}.
\end{align*}
The theorem then follows by taking $\gamma = \Omega(n)$.
\end{proof}

\paragraph{OSP-ness for given verification probabilities.}
Fix $i$ and let $t$ be her true type.
Let us denote $p_{\max}^i$ as $\max_{\bi, b \in D_i} p_{b,t}(\bi)$.
We prove the following result.
\begin{proposition}
\label{prop:fixProb}
	For every social choice function $f$ and verification probabilities $p$ such that $p^i_{t',t}(\bi)\neq 1$ for all $(i,\bi,t,t')$,
	let $\M_p$ be the mechanism with full probabilistic verification that
	requires agents to directly reveal their type concurrently, 
	implements $f$, uses fines $F^i_{t',t}(\bi)\geq\frac{t(\M(t', \bi))-\lmin}{1-p_{\max}^i}$,
	and verification probabilities $p$.
	Then $\M_p$ is OSP.
\end{proposition}
\begin{proof}
	Fix $i$.
	Since $p_{\max}^i\neq 1$, then
	\begin{align*}
	\inf_{\bi} t(\M_p(t, \bi)) & \geq \lmin \\
	& \geq \sup_{\bi} \left\{t(\M_p(t', \bi)) - (1-p_{t',t}(\bi)) \frac{t(\M_p(t', \bi))-\lmin}{1-p_{\max}^i}\right\}\\
	& \geq \sup_{\bi} \left\{\rule{0ex}{3ex} t(\M_p(t', \bi))-(1-p_{t',t}(\bi)) F_{t',t}(\bi)\right\}.\tag*{\qed}
	\end{align*}
	\let\qed\relax
\end{proof}

A close inspection to the proof reveals that our lower bounds on fines are tight, i.e., for any $t, t'$, $i$ and $\bi$ there is no smaller $F_{t',t}^i(\bi)$ that would guarantee OSP. This in particular means that once the probabilities to verify have been set there is not much flexibility in the fines imposed on agents. Lower fines would only be possible whenever the verification step would reveal not just the valuation for the implemented outcome but more information about $t$ and, consequently, more knowledge about the tuple $(t(x))_{x\in \out}$. For example, in facility location, once we learn $c=t(\M(t', \bi))$ through verification, as we know that the location of the facility is $f=\M((t', \bi)$, we can conclude that the type $t$ of the agent belongs to $D'=\{f-c, f+c\}$. When verification allows to know that $t \in D' \subset D$ we say that it is \emph{$D'$-revealing}. Let $D_{-i}$ denote $\bigtimes_{j \neq i} D_j$. Proposition~\ref{prop:fixProb} can be adapted to prove the following one, where $\lmin' = \inf_{t\in D', \bi \in D_{-i}} \{t(\M(t, \bi))\}$.
\begin{proposition}
\label{prop:full_model}
	For every social choice function $f$ and verification probabilities $p$ such that $p^i_{t',t}(\bi)\neq 1$ for all $(i,\bi,t,t')$,
	let $\hat{\M}_p$ be the mechanism with full probabilistic $D'$-revealing verification that
	requires agents to directly reveal their type concurrently, 
	implements $f$, uses fines $\hat{F}^i_{t',t}(\bi)\geq \frac{t(\M(t', \bi))-\lmin'}{1-p_{\max}^i}$,
	and verification probabilities $p$.
	Then $\hat{\M}_p$ is OSP.
\end{proposition}
Since $\lmin' \geq \lmin$, we have that $\hat{F}^i_{t',t}(\bi) \leq F^i_{t',t}(\bi)$, and thus lower fines are sufficient to implement OSP mechanisms with completely revealing verification.

\section{Partial probabilistic verification}
\label{sec:partial}
One limitation of Theorem~\ref{thm:positive} is that the fines have a very high value (linear in the number of agents) which makes their enforceability doubtful. Looking at the proof, if the mechanism were able to verify a constant number of agents,
then the summation would only have a constant number of addends, $\gamma$ would in turn become a constant and, ultimately, the fines would be reduced significantly.
In this section, we thus investigate whether it is possible to obtain OSP mechanisms that verify few agents.

Specifically, we let $V(\b)$ denote the subset of verifiable agents for type profile $\b$.
Such a subset, just like the outcome, can be chosen randomly and can depend on the declaration of the agents.
Note that $V(\b)=n$ in the full probabilistic verification model.
Here, we consider partial probabilistic verification, and we ask how large should $V(\b)$ be in order to have an OSP mechanism.

As in Proposition~\ref{prop:fixProb}, for bounded domains we can guarantee through fines that, no matter the quality of the verification device, truthtelling will be obviously dominant for all the agents in $V(\b)$. Therefore, the mechanism ``only'' needs to obviously incentivize the agents that are not in $V(\b)$.
We next prove the number of these agents needs to be small for any OSP mechanism with partial probabilistic verification, even in the case in which the designer can choose the agents to verify in $V(\b)$. In fact, we next show that for every $\varepsilon > 0$ there is a problem for which every $\varepsilon$-OSP mechanism needs to verify at least $n-o(n)$ agents,
where $n$ is the total number of agents.
We also show that this result continues to hold even in the Bayesian setting and if one requires to solve the problem only asymptotically.
Finally, we prove that this bound is essentially tight: indeed, for every $\varepsilon > 0$ there is an $\varepsilon$-OSP in expectation mechanism that is able to implement asymptotically every social choice function, by verifying at most $n-o(n)$ agents.
Moreover, if the agents may be sometimes \emph{imposed} to adopt a given ``reaction'' to the output of the mechanism, then there is also an OSP in expectation mechanism with the same features.

\paragraph{How many agents must be verified?}
Consider the following problem, known as \emph{public project} problem. We need to decide whether to implement or not a public project (e.g., building a bridge) whose cost is $c$. The society is comprised of a set $N$ of $n$ individuals (also termed agents or customers) that we index with integers from $0$ to $n-1$. The valuation of agent $i$ if the project is implemented may be either $v_i(1) = 1$ or $v_i(1) = \delta > 0$, where, $\delta \ll 1$ (e.g., $\delta = \frac{1}{n^2}$). We say that the type of $i$ is $1$ in the first case, and $\delta$ otherwise. Moreover, each agent has valuation $v_i(0) = 0$ if the project is not implemented.
The designer would like to implement the project only if at least $c$ individuals have type $1$. In other words, the designer would like to
implement the \emph{public project function} $f$ that returns $1$ if $\sum_i v_i(1) \geq c$, and $0$ otherwise.
This has been introduced by \cite{jackson1992implementing} and it is a basic and very well studied problem in economics and computer science (see, e.g., \citep{AE09} and references therein).

{Consider an $\varepsilon$-OSP mechanism $\M$ that implements the public project function $f$ and let $\tp$ be a type profile.
Observe that $t_i$ can assume only two values: $1$ and $\delta$.
Given type $t_i$ we will denote with $\overline{t_i}$ the unique alternative type.
Let $(\mathbf{S}_i(t_i))_{i \in N}$ be the profile of $\varepsilon$-obviously dominant strategies such that $\M((\mathbf{S}_i(t_i))_{i}) = f(\tp)$
if the mechanism is deterministic, and $E_{d_c \sim \delta_c}[\M(d_c,(\mathbf{S}_i(t_i))_{i})] = f(\tp)$ otherwise.
Given $\M$, $\tp$, a realization $d_c$ of $\delta_c$ and strategies $(\mathbf{S}_i(t_i))_i$, let $\rho(i)$ be the first history $h$, if any,
such that $h$ can be reached when the mechanism is played according to $(d_c, (\mathbf{S}_i(t_i))_i)$, $P(h) = i$,
and every action $a \in \mathcal{A}(h) \setminus \mathbf{S}_i(t_i)(h)$ belongs to $\mathbf{S}_i(\overline{t_i})(h)$.
In other words, $\rho(i)$ is the first node, if any, in the path of the arborescence defined by $\M$, $d_c$ and strategies $\mathbf{S}_i(t_i)$,
such that the agent $i$ called to take an action at that node that either signals $t_i$ or signals the alternative type $\overline{t_i}$.
For this reason, we say that agent $i$ \emph{reveals} her type at $\rho(i)$. Note that, in general, $\rho(i)$ might not exist for all the agents $i$, profiles $\tp$, realizations $d_c$ and strategies $(\mathbf{S}_i(t_i))_i$; for these agents, we let $\rho(i)$ be a dummy history $h^*$ such that $h^* \nprec h$ for all histories $h$ defined by $\M$ and no agent takes actions at $h^*$.
We will denote with $n_i$ and $k_i$ the number of agents that revealed their types,
and the ones that declared to have type $1$, before $i$ reveals her type, respectively.
That is, $n_i = \left|\left\{j \colon \rho(j) \prec \rho(i)\right\}\right|$ and $k_i = \left|\left\{j \colon \rho(j) \prec \rho(i) \text{ and } t_j = 1\right\}\right|$.}

We now show a key property of this problem,
that is enjoyed by every $\varepsilon$-OSP mechanism for this problem and for every type profile.
\begin{lemma}
\label{lem:ver_cond}
 For every $\varepsilon \in [0,1)$ {every type profile $\tp$, and every realization $d_c$ of $\delta_c$},
 an $\varepsilon$-OSP mechanism that implements the public project function
 has to verify every customer $i$ with $n_i \leq n+k_i-c-1$ and $k_i \leq c-2$.
 Moreover, if {$i$ takes an action according to $\mathbf{S}_i(1)$ at $\rho(i)$},
 then $i$ must be verified even if $n_i = n+k_i-c$ or $k_i = c-1$.
\end{lemma}
\begin{proof}
	Let us begin by observing that since the mechanism implements the public project function then if $k_i \in \left[c + n_i - n, c-1\right]$ then $\rho(i)\neq h^*$. In words, whenever the threshold $c$ has not yet been reached and there are still enough agents who have not revealed their type, the mechanism will need to keep querying agents in order to establish its outcome.
	
 Let us first suppose that $i$ has type $1$.
%
 If $c + n_i + 1 - n \leq k_i \leq c-2$, then {$\mathbf{S}_i(1)$} is not obviously dominant:
 the worst outcome for $i$ when she plays this strategy is achieved when no remaining individual reveals type $1$,
 so that the mechanism's outcome is 0 (since there are at most $c-1$ customers with type $1$), and this outcome has value $0$ for $i$;
 instead, by taking the action corresponding to {$\mathbf{S}_i(\delta)$}, the best outcome would occur when at least $c-k_i$ among the remaining agents
 (that are $n-n_i-1 \geq c-k_i$) declare type $1$, so that the mechanism returns $1$, by giving to $i$ utility $1 > 0 + {\varepsilon}$.

 Hence, in order for the mechanism to be $\varepsilon$-OSP it is necessary to verify agent $i$ if she
 {takes action according to $\mathbf{S}_i(\delta)$ at $\rho(i)$} whenever $k_i \in \left[c+n_i+1-n, c-2\right]$.
 A similar argument also can be adopted when $i$'s type is $\delta$, from which we achieve that,
 if the mechanism is $\varepsilon$-OSP, then the agent $i$ that {takes action according to $\mathbf{S}_i(1)$ at $\rho(i)$}
 must be verified whenever $k_i \in \left[c + n_i - n, c-1\right]$.
\end{proof}

It is instructive to observe that if $k_i > c - 2$ and $\rho(i) \neq h^*$
then it is obviously dominant for agent $i$ of type $1$ to play according to strategy $\mathbf{S}_i(1)$ at $\rho(i)$. 
Indeed, since the mechanism implements the public project function, it returns $1$ in this case,
and $i$ will achieve her maximum possible utility, i.e., $1$, regardless of the actions of the others. Similarly, if the number of agents different from $i$ whose type is still undefined is not sufficient to reach the threshold $c$,
i.e., $n - n_i -1 < c - k_i$, then it is obviously dominant for $i$ to to play according to strategy $\mathbf{S}_i(1)$,
otherwise she will achieve the minimum possible utility, i.e., $0$, regardless of the actions of the remaining agents. A similar argument holds also for agent $i$ of type $\delta$ whenever $\rho(i) \neq h^*$ and either $k_i > c - 1$ or $n - n_i < c - k_i$. Hence, Lemma~\ref{lem:ver_cond} is in a sense tight, since the customers that it identities are the \emph{only} ones that need to be verified in an $\varepsilon$-OSP mechanism that implements the public project function.

{Observe that $\rho$ defines an order $\pi$ according to which agents reveal their type,
such that $\pi(1)$ contains all agents $i$ such that $\rho(j) \nprec \rho(i)$ for every agent $j$,
$\pi(2)$ contains all agents $i$ such that $\rho(j) \nprec \rho(i)$ for every agent $j$ not in $\pi(1)$, and so on; $\pi$ is arbitrarily completed by the agents $i$ such that $\rho(i)=h^*$. This revelation order may be defined by the deterministic or randomized choices of the mechanism
or by nature (that is, the mechanism processes agents as they come)\footnote{Unless differently specified, we do not make any difference about this feature and our results hold no matter the source of this order.}.
Lemma~\ref{lem:ver_cond} proves that the number of verified agents depends only on this revelation order. For this reason, we will from now on focus only on this order and we will forget the other details of the mechanism.
In particular, we will say that}
the \emph{record} at time $t \geq 1$ is $R_t = ((\pi(1),\bp_{\pi(1)}), \ldots, (\pi(t),\bp_{\pi(t)}))$,
where for a subset of agents $S \subseteq N$, $\bp_S = (b_i)_{i\in S}$.
We will denote with $\calR_t$ the set of all possible records at time $t$ and with $N_\pi^{t-1} = \bigcup_{j < t} \pi(j)$ the set of agents that revealed their type before time $t$.
Moreover we set $\calR_0=R_0 = \emptyset$.
A \emph{selection rule} $\bsig = (\sigma_1, \ldots, \sigma_l)$, $l \leq n$, where each $\sigma_t \colon \calR_{t-1} \rightarrow \Delta(2^N)$, associates to each possible record $R_{t-1}$ a probability distribution $\sigma_t(R_{t-1})$ over the subsets of agents in $N \setminus N_\pi^{t-1}$. Roughly speaking, the selection rule specifies how the mechanism (nature, resp.) selects the next players to reveal their type. This definition allows us to represent every selection rule, even adaptive ones (in which players are selected based also on the type revealed by previous agents).
In what follows, we focus on selection rules defined (defined by nature and faced, respectively) by of an $\varepsilon$-OSP mechanism with bounded approximation ratio.

Given $c$ and a selection rule $\bsig$, for a type profile $\tp$ such that $\left|\left\{i \colon t_i = 1\right\}\right| = c$,
we let $\tau_{\bsig,c}(\tp)$ be the random variable that measures the number of agents that have been verified by the mechanism on type profile $\tp$.
When clear from the context, we omit $c$ from the subscript.

We begin by observing that if a selection rule selects more than one agent to reveal their type at the same time,
then it can only reach the thresholds of Lemma \ref{lem:ver_cond} later than the selection rule wherein the revelation order is serialized.
So, serialization cannot increase the number of verified customers.
Moreover, as highlighted in \citep{BadeG17,mackenzie2017revelation},
serialization does not affect the OSP-ness of the mechanism.
Thus, we have the following observation.
\begin{obs}
 \label{obs:single_better}
  Fix $\varepsilon \in [0,1)$,
  $t > 0$ and a record $R_{t-1}$.
  Let $\bsig$ be a selection rule of an $\varepsilon$-OSP mechanism
  that implements the public project function,
  such that $\sigma_t(R_{t-1})$ assigns non-zero probability to a subset of agents $A$ with $|A| \geq 2$.
  Then there is an alternative $\varepsilon$-OSP mechanism
  that implements the public project function
  with selection rule $\bsig'$
  such that $\sigma'_t(R_{t-1})$ assigns non-zero probability only to singletons of agents,
  and for every type profile $\tp$ it holds that $E[\tau_{\bsig}(\tp)] \geq  E[\tau_{\bsig'}(\tp)]$.
\end{obs}

The \emph{uniform selection} rule $U$ returns, for every $t$ and for every record, the uniform distribution over agents that have not yet revealed their type.  Next we show that for every selection rule $\bsig$ there is an instance $\tp$ in which $\bsig$ performs worse than the uniform selection rule, in terms of the expected number of agents verified $E[\tau_{\bsig}(\tp)]$.

\begin{lemma}
 \label{lem:worst_order}
  For $\varepsilon \in [0,1)$, $c>0$ and selection rule $\bsig$ for an $\varepsilon$-OSP mechanism
  that implements the public project function,
  there is $\tp$ s.t. $\left|\left\{i \colon t_i = 1\right\}\right| = c$ and $E[\tau_{\bsig}(\tp)] \geq  E[\tau_U(\tp)]$.
\end{lemma}
\begin{proof}
 According to Observation~\ref{obs:single_better}, we can assume without loss of generality that at each time step $\bsig$ assigns positive probability only to singletons.

 Let $P$ be the uniform distribution on the type profiles $\tp$ such that $\left|\left\{i \colon t_i = 1\right\}\right| = c$, i.e.
 $P(\tp) = \left(\binom{n}{c} (n-c)! c!\right)$ if $\left|\left\{i \colon t_i = 1\right\}\right| = c$ and $P(\tp) = 0$ otherwise.
 Observe that, for every $t$, every $R_{t-1}$,
 and for agent $i$ selected at time $t$,
 \begin{equation}
  \label{eq:unif_prop}
  \Pr_P\left(t_i = 1 \mid R_{t-1}\right) = \frac{c-k_i}{n-n_i}.
 \end{equation}

 We prove by induction on $i = 1, \ldots, n$
 that $E_{\tp \sim P}[\tau_{\bsig}(\tp_{N \setminus N_\pi^{n-i}}) \mid R_{n-i}] \geq E_{\tp \sim P}[\tau_U(\tp_{N \setminus N_\pi^{n-i}}) \mid R_{n-i}]$
for every $R_{n-i}$ (and thus $\tp_{N_\pi^{n-i}}$).
 Then, since $R_0 = \emptyset$, we have that $E_{\tp \sim P}[\tau_{\bsig}(\tp)] \geq  E_{\tp \sim P}[\tau_U(\tp)]$.
 The lemma follows because it must exist $\tp^\star$ such that $\left|\left\{i \colon t^\star_i = 1\right\}\right| = c$
 and $E[\tau_{\bsig}(\tp^\star)] \geq  E[\tau_U(\tp^\star)]$.

 The base case is $i = 1$. Observe that for every selection rule exactly the same agent will reveal her type,
 namely, the only one that has not revealed her type in $R_{n-1}$.
 For the given record $R_{n-1} \in \calR_{n-1}$, let $k$ denote the number of agents who declared $1$.
 It must be that either $k = c$ or $k=c-1$ (since we are only considering profiles with exactly $c$ agents with type $1$).
 According to Lemma~\ref{lem:ver_cond}, in the first case the last agent must not be verified,
 whereas in the last case it is necessary to verify the last agent only if she declares $1$.
 However, since this choice is independent from the selection rule and from $\tp_{N_\pi^{n-1}}$, the claim trivially holds.

 Assume that the claim holds for $i-1$ (i.e., for every $R_{n-i+1}$ we have that
 $E_{\tp \sim P}[\tau_{\bsig}(\tp_{N \setminus N_\pi^{n-i+1}}) \mid R_{n-i+1}] \geq  E_{\tp \sim P}[\tau_U(\tp_{N \setminus N_\pi^{n-i+1}}) \mid R_{n-i+1}]$).
 We prove it also for $i$.

 Consider a record $R_{n-i}$.
 If at least $c$ agents declaring type $1$ have revealed her type, or the number of those whose type is still unknown is too low for reaching the threshold $c$, then no further verification needs to be made regardless of the selection rule adopted in the remaining steps. Hence, the claim trivially holds.
 If instead $k \leq c-1$ customers declared $1$ and there are exactly $c-k$ agents whose type is unknown,
 then for every $\tp$ such that $\left|\left\{i \colon t_i = 1\right\}\right| = c$, it will occur
 that all remaining agents have type $1$ and thus, according to Lemma~\ref{lem:ver_cond} they must be all verified,
 regardless of the order in which they are selected.
 Hence, the claim trivially holds.
 Consider now the case that $k \leq c-1$ customers declared $1$ and there are $p \geq c-k+1$ agents that have not yet revealed her type.
 If $\sigma_{n-i+1}(R_{n-i}) = U(R_{n-i})$, then the claim directly follows from the inductive hypothesis.
 Otherwise, let $S=N \setminus N_{\pi}^{n-i}$.
 Since, by Lemma~\ref{lem:ver_cond} the agent selected at the $(n-i+1)$-th time step will be surely verified,
 the expected number of agents verified is
 \begin{align*}
  & E_{\tp \sim P}[\tau_{\bsig}(\tp_{N \setminus N_{\pi}^{n-i}}) \mid R_{n-i}] = 1 + \sum_{z \in S} \sigma_{n-i+1}(R_{n-i})(z)\\
  & \qquad \cdot \sum_{\beta \in \{\delta, 1\}} \Pr_P(b_z = \beta) E_{\tp \sim P}[\tau_{\bsig}(\tp_{N \setminus N_{\pi}^{n-i}}) \mid R_{n-i} \cup (z,\beta)].
 \end{align*}
 By using that the mechanism is $\varepsilon$-OSP, we achieve that
 \begin{align*}
  & E_{\tp \sim P}[\tau_{\bsig}(\tp_{N \setminus N_{\pi}^{n-i}}) \mid R_{n-i}] = 1 + \sum_{z \in S} \sigma_{n-i+1}(R_{n-i})(z)\\
  & \qquad \cdot \sum_{\beta \in \{\delta, 1\}} \Pr_P(t_z = \beta) E_{\tp \sim P}[\tau_{\bsig}(\tp_{N \setminus N_{\pi}^{n-i}}) \mid R_{n-i} \cup (z,\beta)].
 \end{align*}
 Moreover, from the inductive hypothesis, it follows that
 \begin{align*}
  & E_{\tp \sim P}[\tau_{\bsig}(\tp_{N \setminus N_{\pi}^{n-i}}) \mid R_{n-i}] \geq 1 + \sum_{z \in S} \sigma_{n-i+1}(R_{n-i})(z)\\
  & \qquad \cdot \sum_{\beta \in \{\delta, 1\}} \Pr_P(t_z = \beta) E_{\tp \sim P}[\tau_{U}(\tp_{N \setminus N_{\pi}^{n-i}}) \mid R_{n-i} \cup (z,\beta)].
 \end{align*}

 Observe that, by anonymity of the uniform selection rule, $E[\tau_{U}(\tp) \mid R_{t}]$ depends only on the number of agents
 that declared type $1$ in $R_t$ and not on their identities. That is, $E[\tau_{U}(\tp) \mid R_{t}] = E[\tau_{U}(\tp) \mid \kappa]$,
 if in $R_t$ only $\kappa$ agents declared $1$.
 Hence, and according to \eqref{eq:unif_prop}, we have
 \begin{align*}
 & E_{\tp}[\tau_{\bsig}(\tp_{N \setminus N_{\pi}^{n-i}}) \mid R_{n-i}] = 1 +\\
 & \qquad \left(\frac{p-c+k}{p} E_{\tp}[\tau_{U}(\tp_{N \setminus N_{\pi}^{n-i}}) \mid k]+\right.\\
 & \qquad \left.\frac{c-k}{p} E_{\tp}[\tau_{U}(\tp_{N \setminus N_{\pi}^{n-i}}) \mid k+1]\right) \sum_{z \in S} \sigma_{n-i+1}(R_{n-i})(z).
 \end{align*}
 Then $\sum_{z \in S} \sigma_{n-i+1}(R_{n-i})(z) = 1 = \sum_{z \in S} U(z)$ implies $E_{\tp \sim P}[\tau_{\bsig}(\tp_{N \setminus N_{\pi}^{n-i}}) \mid R_{n-i}]=E_{\tp \sim P}[\tau_{U}(\tp_{N \setminus N_{\pi}^{n-i}}) \mid R_{n-i}]$.
\end{proof}

We can now state the main theorem of this section.
\begin{theorem}
\label{thm:opt_noBayes}
 For all $\varepsilon$-OSP mechanisms implementing the public project function, with $\varepsilon \in [0,1)$,
 there is an instance on which the mechanism verifies in expectation $n-o(n)$ agents.
\end{theorem}
\begin{proof}
 Let $c = 1+\sqrt{n-1}$. We next show that if the mechanism adopts the uniform selection rule,
 then for every $\tp$ such that $\left|\left\{i \colon t_i = 1\right\}\right| = c$,
 $E[\tau_U(\tp)] \geq n-o(n)$. The theorem then follows by merging this result with Lemma~\ref{lem:worst_order}.
 Observe indeed that
 \begin{align*}
  E[\tau_U(\tp)] & \geq (n-c-1) \cdot \left(1 - \Pr(\tau_U(\tp) < n-c-1)\right)\\
  & = (n-o(n)) \left(1 - \Pr(\tau_U(\tp) < n-c-1)\right).
 \end{align*}
 We next prove that $\Pr(\tau_U(\tp) < C) = o(1)$, concluding in this way the proof.
For $\tau_U(\tp)$ to be less than $n-c-1$ it must be the case that $c-1$ customers with value $1$
 have been selected among the first $n-c-2$ agents that revealed her type.
 Indeed, since a mechanism that implements the public project function must reveal the type of all agents
 until either $c$ agents have declared to have type $1$ or the number of agents whose type is unknown is insufficient to reach this threshold,
 then the $(n-c-1)$-th agent must necessarily reveal her type, and by Lemma~\ref{lem:ver_cond} it must be verified.
 Observe that there are $\binom{n-c-2}{c-1}$ revealing orders so that this property is satisfied
 among the $\binom{n}{c}$ orders in which the $c$ customers with type $1$ can reveal it.
 Since we are using the uniform selection rule, these arrangements have the same probability, and thus
 \begin{align*}
  & \Pr(\tau_U(\tp) < n-c-1) = \frac{\binom{n-c-2}{c-1}}{\binom{n}{c}}\\
  & = \frac{(n-c-2)!}{(c-1)!(n-2c+1)!} \cdot \frac{c!(n-c)!}{n!}\\
  & = c \cdot \frac{(n-c-2)!}{(n-2c+1)!} \cdot \frac{(n-c)!}{n!}\\
  & = c \cdot \frac{\prod_{i = 1}^{c-1} (n-c+1-i)}{\prod_{i = 0}^{c-1} (n-i)}\\
  & = \frac{c}{n} \cdot \prod_{i = 1}^{c-1} \frac{n-c+1-i}{n-i} = \frac{c}{n} \cdot \prod_{i = 1}^{c-1} \left(1 - \frac{c-1}{n-i}\right)\\
  & \leq \frac{c}{n} \left(1 - \frac{c-1}{n-1}\right)^{c-1}\\
  & = \frac{1 + \sqrt{n-1}}{n} \left(1 - \frac{1}{\sqrt{n-1}}\right)^{\sqrt{n-1}}\\
  & \leq \frac{1 + \sqrt{n-1}}{en} = o(1). \tag*{\qed}
 \end{align*}
 \let\qed\relax
\end{proof}
We will next show that this theorem can be extended in order to work even if we have weaker constraints.
Specifically, we prove that the result holds even in the Bayesian setting (Proposition~\ref{prop:bayes}),
i.e., if the mechanism designer knows the distribution from which the agents' types are drawn,
or it holds if we only require that the mechanism implements the public project function only asymptotically (Proposition~\ref{propt:opt_noBayes_asym}).

We focus first on the Bayesian setting.
\begin{proposition}
 \label{prop:bayes}
 For every $\varepsilon$-OSP mechanism implementing the public project function in the Bayesian setting, with $\varepsilon \in [0,1)$, there is an instance for which the mechanism verifies in expectation $n-o(n)$ agents.
\end{proposition}
\begin{proof}
 Consider an instance in which each agent has the same probability $p$ of having type $1$.
 By symmetry of agents, the order $\pi$ in which agent reveal their type cannot depend on this prior.
 Moreover, since the mechanism implements the public project function,
 each agent must reveal her type at some time as long as either $c$ agents have declared to have type $1$
 or the number of agents whose type is unknown is insufficient to reach this threshold.

 Without loss of generality, we rename the agents so that an agent with lower index is scheduled before an agent with higher index.
 That is, $i < j$ implies that $n_i \leq n_j$. It then follows that $n_i \leq i$.

 Let us define the random variable $X_i$ that assumes value $1$ if the customer $i$ is verified, and $0$ otherwise.
 Observe that $X = \sum_i X_i$ is the number of verified customers.
 By Lemma \ref{lem:ver_cond}, the expected number of verified agents is
\begin{align*}
 E\left[X\right] & = \sum_{i=0}^{n-1} \Pr(X_i = 1)\\
 & \geq \sum_{i=0}^{n-1} \Pr\left(k_i \in \left[c+n_i+1-n, c-2\right]\right).
\end{align*}
If $i \leq n-c-1$, then $c+n_i+1-n \leq 0$. Hence, we have $\Pr\left(k_i \in \left[c+n_i+1-n, c-2\right]\right) = \Pr\left(k_i \leq c-2\right)$.
Since $\Pr\left(k_i \leq c-2\right) = 1$ if $n_i \leq c-2$,
and thus for at least the first $c-1$ customers that revealed her type,
we then have
\begin{equation}
\label{eq:ex_ver}
\begin{aligned}
 E\left[X\right] & \geq \sum_{i=0}^{n-c-1} \Pr\left(k_i \leq c-2\right)\\
 & \geq c-1 + \sum_{i=c}^{n-c-1} \Pr\left(k_i \leq c-2 \mid n_i \geq c-1\right)
 \end{aligned}
\end{equation}
Hence, we achieve that
\begin{equation}
\label{eq:p_ver}
\begin{aligned}
 & \Pr\left(k_i \leq c-2 \mid n_i \geq c-1\right) =\\
 & \quad \sum_{j = 0}^{c-2} \Pr\left(k_i = j \mid n_i \geq c-1\right) =\\
 & \quad \sum_{j = 0}^{c-2} \binom{n_i}{j} p^j (1-p)^{n_i-j} =\\
 & \quad (1-p)^{n_i} \cdot \sum_{j = 0}^{c-2} \frac{1}{j!} \frac{n_i!}{(n_i - j)!} \left(\frac{p}{1-p}\right)^j =\\
 & \quad (1-p)^{n_i} \cdot \left(1 + \sum_{j = 1}^{c-2} \frac{1}{j!} \frac{n_i!}{(n_i - j)!} \left(\frac{p}{1-p}\right)^j\right) \geq\\
 & \quad (1-p)^{n_i}.
\end{aligned}
\end{equation}
Let $\rho_i = \frac{1}{p}$. It is well-known that $\left(1 + \frac{1}{x}\right)^x < e < \left(1 + \frac{1}{x-1}\right)^x$,
from which we have that
$$
 (1-p)^{n_i} = \left(1-\frac{1}{\rho_i}\right)^{n_i} > e^{-\frac{n_i}{\rho_i - 1}} > 1 - \frac{n_i}{\rho_i - 1} = 1 - \frac{n_i p}{1 - p}.
$$
If $p = \frac{g}{g+(n-c-1)^2} \leq \frac{g}{g+n_i(n-c-1)}$ for some $g=o(n)$, then it is easy to see that
\begin{equation}
 \label{eq:lb_p_ver}
 (1-p)^{n_i} > 1 - \frac{n_i p}{1 - p} \geq 1 - \frac{g}{n-c-1}.
\end{equation}
By putting together \eqref{eq:ex_ver}, \eqref{eq:p_ver}, and \eqref{eq:lb_p_ver},
we can conclude that the mechanism must verify at least $n-c-1-o(n)$ customers.
The theorem then follows by taking $c=o(n)$.
\end{proof}

We now extend Theorem~\ref{thm:opt_noBayes} in order to work even with mechanisms that implement the public project problem only asymptotically.
\begin{proposition}
\label{propt:opt_noBayes_asym}
 For every $\varepsilon \in [0,1)$ and for every $\varepsilon$-OSP mechanism that implements the public project function asymptotically,
 there is an instance for which the mechanism must verify in expectation $n-o(n)$ agents.
\end{proposition}
\begin{proof}
 Let $c = 1+\sqrt{n-1}$. We next show that the probability that the mechanisms verifies at most $n-c-1$ agents is $o(1)$.
 Note that, differently from the case of mechanisms implementing exactly the desired social choice function,
 there may be two reasons for which a non-exact mechanism verifies at most $n-c-1$ agents:
 either $c-1$ with type $1$ are found among the first $n-c-1$ agents, or the mechanism reveals the type of less than $n-c$ of them.
 Theorem~\ref{thm:opt_noBayes} proves that the first event occurs only with negligible property.
 We next show that this is the case even for the second one.

 As above, let $X$ be the random variable that measures the number of verified agents.
 Moreover, let $E_i$ be the event in which the mechanism has not seen $c$ agents with type $1$ after at least $i+1$ agents that revealed the type.
 Then,
 \begin{align*}
  & \Pr(X \leq n-c-1) =\\
  & \quad \Pr(X \leq n-c-1 \mid \bar{E}_{n-c-1})\cdot\Pr(\bar{E}_{n-c-1})\\
  & \quad + \Pr(X \leq n-c-1 \mid {E}_{n-c-1})\cdot\Pr({E}_{n-c-1})\\
  & \leq o(1) + \Pr(X \leq n-c-1 \mid {E}_{n-c-1}),
 \end{align*}
 where the inequality follows from $\Pr(\bar{E}_{n-c-1}) = o(1)$, as shown in the proof of Theorem~\ref{thm:opt_noBayes}.
 The theorem follows by proving that $\Pr(X \leq n-c-1 \mid {E}_{n-c-1}) = o(1)$.

 Suppose that this is not the case, and $\Pr(X \leq n-c-1 \mid {E}_{n-c-1}) = \Omega(1)$.
 Note that by definition of ${E}_{n-c-1}$, the mechanism terminates before knowing whether there are at least $c$ agents with type $1$.
 Let us suppose that the mechanism decides to not implement the project in this case.
 Consider then an instance $\tp$ with exactly $c$ agents with type $1$.
 Observe that in this case the mechanism does not implement the project whenever it stops after $n-c-1$ agents revealed the type,
and the $c$ agents with type $1$ are not found among these.
 That is, the mechanism returns the wrong outcome with probability at least
 $\Pr(X \leq n-c-1 \mid {E}_{n-c-1})\cdot\Pr({E}_{n-c-1}) \geq \Omega(1) \cdot (1-o(1)) = \Omega(1).$
 But this contradicts the assumption that the mechanism implements the public project function asymptotically.
\end{proof}

\paragraph{Tightness of the bound.}
For a social choice function $f$, the \emph{sensitivity} of $f$ is the smallest integer $d$
such that
for every $i$, every $t_i, t'_i \in D_i$, every $\tp_{-i} \in D_{-i}$, and every $s$.
$$
 \left|f((t_i,\tp_{-i}), s) - f((t'_i,\tp_{-i}), s)\right| \leq \frac{d}{n}.
$$

Fix $c = o(n)$ and $\varepsilon > 0$ such that $\frac{4dc}{\varepsilon} = o(n)$ and let $\beta = \frac{n\varepsilon}{2dc}$.

Consider the \emph{$\beta$-Exponential Mechanism with Partial Verification} {$\M^\beta$}:
ask agents to reveal their type in sequential order;
given a declaration profile $\bp \in D$, choose the outcome ${s_\ell} \in \out {= \{s_1, \ldots, s_{|\out|}\}}$ according to the probability distribution
\begin{equation}
\label{eq:M_beta}
 {M^\beta_\bp(s_\ell)} = \frac{e^{\beta f(\bp,{s_\ell})}}{\sum_{z \in \out} e^{\beta f(\bp,z)}};\footnote{This is equivalent to saying that the mechanism chooses $d_c$ according to the distribution $U(0,1]$, i.e., the uniform distribution on $(0,1]$, and returns $s_\ell$ whenever $d_c \in \left(\sum_{j < \ell} M^\beta_\bp(s_j), \sum_{j < \ell} M^\beta_\bp(s_j) + M^\beta_\bp(s_\ell)\right]$. This formulation would be coherent with the terminology adopted throughout the paper but undoubtedly be more complex.}
\end{equation}
for the first $n-c$ agents, verify if their declared type $b_i$ coincides with their real type $t_i$.

We have the following theorem.
\begin{theorem}
\label{thm:diff}
 For every social choice function $f$, if $|S| \leq e^{o(n)}$, then the $\beta$-Exponential Mechanism with Partial Verification is $2\varepsilon$-OSP {in expectation} and it implements $f$ asymptotically.
\end{theorem}
The proof follows from the following lemmas.
\begin{lemma}
\label{lem:2epsOSP}
 The $\beta$-Exponential Mechanism with Partial Verification is $2\varepsilon$-OSP {in expectation}.
\end{lemma}
The proof mimics Lemma~1 and Lemma~2 in \citep{diff}.
\begin{proof}
 The claim obviously hold for every agent selected among the first $n-c$, since an untruthful bid would be discovered by the verification; as from above, we can then use fines to obviously discourage this misreport.

 We will next show that for every agent $i$ not selected among the first $c$, any type $t_i$,
 and any two bid profiles $\b=(b_i,\bp_{-i})$ and $\b'=(t_i,\bp'_{-i})$ such that $b_j = b'_j = t_j$ for each agent $j \neq i$ among the first $n-c$ agents, it holds that
 $$E_{{s \sim M^\beta_{\b}}}[v_i(t_i,s)] - E_{{s \sim M^\beta_{\b'}}}[v_i(t_is)] \leq 2\varepsilon.\footnote{Observe that, since the random choice of the mechanism is only done at a terminal node, then for every agent $i$, for every pair of strategies $S_i = \textbf{S}_i(b_i)=b_i, S'_i = \textbf{S}_i(b'_i)=b_i'$, for every $I_i \in \alpha(S_i, S'_i)$ and every $S_{-i} = (\textbf{S}_j(b_j))_{j \neq i} \in \widehat{W}(I_i, S_i, S'_i)$, we have that $(S_{-i}, d_c)$ witnesses that $I_i \in \alpha(S_i, S'_i)$ for every realization $d_c$. In other words, $\delta_c \mid (I_i, S_i, S_i', S_{-i}) = \delta_c$, and thus $$E_{d_c \sim \delta_c \mid (I_i, S_i, S_i', S_{-i})}[v_i(t_i, \M^\beta(d_c, \b))] = E_{d_c \sim U(0,1]}[v_i(t_i, \M^\beta(d_c,\b))] = E_{s \sim M^{\beta}_{\b}}[v_i(t_i, s)]$$ for every $b_i, b'_i$ and $\bp_{-i}$.}$$
 The lemma then follows.

 To this aim, let us first prove that the probability that $M^\beta$ returns $s$ does not change too much, regardless the bid profile in input.
 Indeed, since $f$ is $d$-sensitive and $(b_i,\bp_{-i})$ and $(t_i,\bp'_{-i})$ differ in at most $c$ positions,
 then $f((b_i,\bp_{-i}),z) - \frac{cd}{n} \leq f((t_i,\bp'_{-i}),z) \leq f((b_i,\bp_{-i}),z) + \frac{cd}{n}$ for every $z \in \out$. Hence,
 \begin{align*}
  \frac{M^\beta_{{\b}}(s)}{M^\beta_{{\b'}}(s)} & = \frac{\frac{e^{\beta f((b_i,\bp_{-i}),s)}}{\sum_{z \in S} e^{\beta f((b_i,\bp_{-i}),z)}}}{\frac{e^{\beta f((t_i,\bp'_{-i}),s)}}{\sum_{z \in S} e^{\beta f((t_i,\bp'_{-i}),z)}}}\\
  & \leq \frac{\frac{e^{\beta f((b_i,\bp_{-i}),s)}}{\sum_{z \in S} e^{\beta f((b_i,\bp_{-i}),z)}}}{\frac{e^{\beta(f((b_i,\bp_{-i}),s)-\frac{cd}{n})}}{\sum_{z \in S} e^{\beta(f((b_i,\bp_{-i}),z)+\frac{cd}{n})}}}
  \leq e^{\varepsilon},
 \end{align*}
 where the last inequality follows from our choice of $\beta$.

 Hence, we have that
 \begin{align*}
  E_{{s \sim M^\beta_{\b}}}[v_i(t_i,s)] & = \sum_{s \in \out} v_i(t_i,s) M^\beta_{{\b}}(s)\\
  & \leq e^{\varepsilon} \cdot \sum_{s \in \out} v_i(t_i,s) M^\beta_{{\b'}}(s)\\
  & = e^{\varepsilon} \cdot E_{{s \sim M^\beta_{\b'}}}[v_i(t_i,s)].
 \end{align*}
 Then
 \begin{align*}
  & E_{{s \sim M^\beta_{\b}}}[v_i(t_i,s)] - E_{{s \sim M^\beta_{\b'}}}[v_i(t_i,s)]\\
  & \quad \leq (e^{\varepsilon} - 1) E_{{s \sim M^\beta_{\b'}}}[v_i(t_i,s)]\\
  & \quad \leq e^{\varepsilon} - 1 \leq 2\varepsilon,
 \end{align*}
 where the last inequalities follows because $v_i(t_i,s) \in [0,1]$ for every $s$ and $e^{\varepsilon} - 1 \leq 2\varepsilon$ for every $\varepsilon \in [0,1]$.
\end{proof}

\begin{lemma}
 \label{lem:approx}
 If $|\out| \leq e^{o(n)}$, then the $\beta$-Exponential Mechanism with Partial Verification implements $f$ asymptotically.
\end{lemma}
The proof mimics Lemma~3 in \cite{diff}.
\begin{proof}
 Consider $n$ such that $\beta|S| > e$ (i.e., $n > \frac{2edc}{\varepsilon |S|}$) and let $\delta = \frac{\log \beta|S|}{\beta}$.
 Observe that, by our choice of $n$, $\log \beta|S| > \log e = 1$, and, consequently, we have that $\delta \geq \frac{1}{\beta} > 0$.

 For every type profile $\tp$ and every \emph{bad} outcome $s \in \tilde{S} = \left\{\tilde{s} \in S \colon f(\tp,s) < \max_z f(\tp, z) - \delta\right\}$, we have that
 $$
  M^\beta_{{\tp}}(s) = \frac{e^{\beta f(\tp,s)}}{\sum_{z \in S} e^{\beta f(\tp,z)}} \leq \frac{e^{\beta\left(\max_z f(\tp, z) - \delta\right)}}{e^{\beta \max_z f(\tp, z)}} = e^{-\beta\delta}.
 $$
 Thus, $M^\beta_{{\tp}}(\tilde{S}) = \sum_{s\in \tilde{S}} M^\beta_{{\tp}}(s) \leq |\tilde{S}| e^{-\beta\delta} \leq |S| e^{-\beta\delta}$.
 In turn, this implies that with probability at least $1 - |S| e^{-\beta\delta}$ the mechanism returns an outcome $s$ such that $f(s) \geq \max_z f(\tp, z) - \delta$. Then,
 \begin{align*}
  E_{{s \sim M^\beta_\tp}}[f(\tp, {s})] & \geq \left(\max_z f(\tp, z) - \delta\right)\left(1 - |S| e^{-\beta\delta}\right)\\
  & \geq \max_z f(\tp, z) - \delta - |S| e^{-\beta\delta},
 \end{align*}
 where in the last inequality we used that $f(\tp, z) \in [0,1]$. By substituting $\delta$ in the above equation, we have
 \begin{align*}
  E_{{s \sim M^\beta_\tp}}[f(\tp, {s})] & \geq \max_z f(\tp, z) - \frac{\log \beta|S|}{\beta} - \frac{1}{\beta}\\
  & \geq \max_z f(\tp, z) - \frac{2\log \beta|S|}{\beta}.
 \end{align*}
 Thus, the mechanism $\M^\beta$ has, in expectation, an additive approximation of $\frac{2\log \beta|S|}{\beta} = \frac{4dc}{n\varepsilon} \log \frac{n \varepsilon |S|}{2dc}$,
 that is $o_n(1)$ since $\frac{4dc}{\varepsilon} = o(n)$ and $|S| \leq e^{o(n)}$.
\end{proof}

We next show that
we can extend, under opportune conditions, the theorem above
to prove the existence of a \emph{strictly} OSP {in expectation} mechanism
that implements asymptotically every social choice function $f$ whenever the mechanism is able to limit the way an agent reacts to the announced outcome.

\paragraph{Tight OSP {in expectation} imposing mechanism.}
Note that every direct-revelation mechanism can be phrased as follows:
First, the designer collects bids $\bp$ from agents.
Then, it proposes an outcome $s$ to agents.
Finally, agents react to the proposed outcome.
E.g., in an auction, when the auctioneer proposes a price, the agent may accept to buy the good at that price, or refuse.
Similarly, in facility location, the auctioneer decides the location of facilities,
whereas the agent reacts by choosing the closest facility.
The utility $v_i$ of the player $i$ will then depend on her type $t_i$, the proposed outcome $s$,
and the \emph{reaction} $r_i$. I.e., $v_i(t_i, s) = v_i(t_i, s, r_i)$.
A reaction $r_i^*(t_i,s)$ is said to be optimal for agent $i$ of type $t_i$ at $s$ if it maximizes the valuation of $i$
when the proposed outcome is $s$. I.e., $r_i^*(t_i,s) = \arg \max_{r} v_i(t_i,s,r)$.
A mechanism is \emph{imposing} if it forces the reaction of agent $i$ declaring $b_i$ to be $r_i^*(b_i,s)$,
i.e., the optimal reaction according to the declared bid.
For example, an auction is imposing if the agent cannot refuse to buy the item if the proposed price is below her declaration;
a mechanism for facility location is imposing if it does not allow an agent to use a facility different
from the one that is closest to agent's declared position.
Hence, in an imposing mechanism, if the bid of agent $i$ is $b_i$,
then $v_i(t_i,s)=v_i(t_i, s, r_i^*(b_i,s))$ for every $i$, $t_i$, $s$.
Moreover, observe that for every $i$, every $t_i$, and every $s$, by definition of optimal reaction, we have that
$v_i(t_i,s,r_i^*(t_i,s)) \geq v_i(t_i,s,r_i^*(b_i,s))$.

We will next prove that, for every social choice function $f$,
there is an imposing mechanism with partial probabilistic verification that implements $f$ asymptotically
and is strictly OSP  {in expectation}.
To this aim,
set $\gamma(i, t_i, b_i) = \max_s v_i(t_i,s,r_i^*(t_i,s)) - v_i(t_i,s,r_i^*(b_i,s))$,
and denote the outcome achieving this maximum as
$s(i, t_i, b_i) = \arg \max_s v_i(t_i,s,r_i^*(t_i,s)) - v_i(t_i,s,r_i^*(b_i,s))$.
Moreover, set $\gamma = \min_{i, t_i, b_i \neq t_i} \gamma(i, t_i, b_i)$.
If $\gamma = 0$, then it means that reactions do not influence the players utility,
and hence imposing a reaction cannot have effect on the mechanism.
For this reason, we will consider henceforth on a setting $(n, D, \out, (v_i)_{i \in [n]})$
for which $\gamma > 0$.

%
Consider a social choice function $f$ with sensitivity $d$ and let $\gamma$ as defined above.
Let $\varepsilon=\sqrt{\frac{\gamma d}{|\out|n}} \sqrt{\log \frac{n\gamma}{2d}}$, $c=o\left(\sqrt{\frac{n}{|\out|\log n}}\right) \geq 1$,
and $q = \frac{2|\out|\varepsilon}{\gamma}$.
Let $P$ be the probability distribution that assigns probability $\frac{1}{|\out|}$ to each $s \in \out$.
We define the imposing mechanism $\M^{\beta,P}$ as follows:
ask agents to reveal the type in sequential order; for each agent $i$ in the first $n-c$ agents, we verify if her declared type $b_i$ coincides with her real type $t_i$;
given a declaration profile $\bp \in D$, then, with probability $1-q$ choose the outcome $s \in \out$ according to probability distribution {$M^\beta_\bp$}
as defined in $\eqref{eq:M_beta}$, with $\beta = \frac{n\varepsilon}{2dc}$,
and with the remaining probability choose the outcome according to probability distribution $P$.

Let $n_0$ be the smallest integer such that $n_0 \geq \frac{8d|\out|}{\gamma} \log \frac{\gamma}{2d}$ and $\frac{n_0}{\log n_0} > \frac{8d|\out|}{\gamma}$.
We the have the following theorem, whose proof mimics Theorem~2 in \citep{diff}
\begin{theorem}
 For every social choice function $f$, if $n>n_0$ and $|\out| = o(n)$, then $\M^{\beta,P}$ is OSP {in expectation}, it verifies at most $n-o(n)$ agents and it implements $f$ asymptotically.
\end{theorem}
\begin{proof}
 The fact that $\M^{\beta,P}$ verifies at most $n-o(n)$ agents follows directly from the definition and our choice for $c$.
 Moreover, in \cite[Lemma~5]{diff}, it is proved that if $n>n_0$, then $\frac{q\gamma}{|S|} \geq 2\varepsilon$.
 Therefore, from Lemma~\ref{lem:2epsOSP} and since
 $$E_{s \sim P}[v_i(t_i,s,r_i^*(t_i,s))] \geq E_{s \sim P}[v_i(t_i,s,r_i^*(b_i,a))] + \frac{\gamma}{|\out|},$$
 it follows that $\M^{\beta,P}$ is strictly OSP {in expectation}.
 Finally, {let $M^{\beta, P}_\tp = (1-q)M^\beta_\tp + qP$ be the distribution from which $\M^{\beta, P}$ draws the outcome. We then} have that
 $$
  E_{{s \sim M^{\beta,P}_\tp}}[f(\tp,{s})] = (1-q) E_{{s \sim M^\beta_\tp}}[f(\tp, {s})] + q E_{{s \sim P}}[f(\tp, {s})] \geq E_{{s \sim M^\beta_\tp}}[f(\tp, {s})].
 $$
 As shown in the proof of Lemma~\ref{lem:approx}, we have that
 $$
  E_{{s \sim M^\beta_\tp}}[f(\tp, {s})] \geq \max_z f(\tp, z) - \frac{4dc}{n\varepsilon} \log \frac{n \varepsilon |S|}{2dc} \geq \max_z f(\tp, z) - \frac{4dc}{n\varepsilon} \log \frac{n \varepsilon |S|}{2d}.
 $$
 Moreover, since $\max_z f(\tp, z) \leq 1$ and $q = \frac{2|S|\varepsilon}{\gamma} < 1$, we have that
 $$
  E_{{s \sim M^\beta_\tp}}[f(\tp, {s})] \geq \max_z f(\tp, z) - q - \frac{4dc}{n\varepsilon} \log \frac{n \varepsilon |S|}{2d} \geq \max_z f(\tp, z) - \frac{2|S|\varepsilon}{\gamma} - \frac{4dc}{n\varepsilon} \log \frac{n \gamma}{2d}.
 $$
 By substituting $\varepsilon$ in the above equation, we have
 $$
  E_{{s \sim M^\beta_\tp}}[f(\tp, {s})] \geq \max_z f(\tp, z) - (4c+2)\sqrt{\frac{d|S|}{\gamma n}}\sqrt{\log \frac{n\gamma}{2d}} = \max_z f(\tp, z) - o_n(1),
 $$
 where the last step follows from our hypothesis on $|S|$ and our choice for $c$.
\end{proof}

\section{Conclusions}
\cite{li2015obviously} proved that OSP is the ``right'' definition of truthfulness for a kind of
``bounded rational'' agents, where the kind of bounded rationality 
(i.e., those who have limited who lack contingent reasoning skills)
is exactly the one observed in many experimental settings.
This makes it interesting to investigate what can and cannot be done with these partially rational agents,
and how these limits can be addressed (e.g., can verification help?).
This motivates ours and most of recent work on the topic \citep{BadeG17,FerraioliV17}.
It would still be interesting to investigate mechanism design for other (possibly, less stringent) notions of bounded rationality.

In Section~\ref{sec:partial}, we focused on the binary public project problem.
The simplicity of this problem, in our opinion, has two advantages: firstly,
it makes our result stronger, since we are giving a negative result;
secondly, it improves the readability of the proof, since we do not need to deal with the complexities of the problem.
However, a detailed look at our proof highlights that we 
uses the structure of the problem only to prove that:
(i) one needs to verify almost every agent that reveal her type until you find a solution;
(ii) there is an instance for which it is highly unlikely to find a solution after only few agents revealed her type.
It is not hard then to see that these properties are enjoyed not only by the public project problem,
but also by its combinatorial counterpart
\citep{buchfuhrer2010computation,dughmi2011truthful}
and many other different problems (e.g., facility location).
It would, however, be interesting to find settings in which an OSP mechanism with partial probabilistic verification exists
that verifies only few agents.

Finally, we proved that an $\varepsilon$-OSP {in expectation} mechanism with partial probabilistic verification that implements a social choice function $f$ asymptotically and verifies at most $n-o(n)$ agents always exists.
We also proved that we can also turn the mechanism above in one that is strictly OSP {in expectation} if we focus on imposing mechanism.
However, it would be interesting to understand: {(i)} whether reactions' imposition is necessary, or whether there is a non-imposing OSP mechanism with the desired properties; {(ii) whether $\varepsilon$-OSP \emph{deterministic} mechanism exists with the desired properties}.

\bibliographystyle{plainnat}
\bibliography{osp}

\end{document}